\documentclass{ifacconf}

\usepackage{graphicx}      % include this line if your document contains figures
\usepackage{natbib}        % required for bibliography       % for matrix
\usepackage{mathtools}
\usepackage{amssymb}
\usepackage{algorithm}
\usepackage{algpseudocode}
\usepackage{booktabs}
\usepackage{tabularx}
\usepackage{url}
\usepackage{xcolor}

\begin{document}
\begin{frontmatter}

\title{Knapsack-based Online Sensor Selection for Vehicle State Estimation\thanksref{footnoteinfo}} 

% Title, preferably not more than 10 words.
\thanks[footnoteinfo]{\copyright~2026 the authors. This work has been accepted to IFAC for publication under a Creative Commons Licence CC-BY-NC-ND. This manuscript is an extended version with additional details.}

\author[First]{Jehyeop Han} 
\author[First]{Minhee Kang} 
\author[Second]{Alessandro Colombo}
\author[Second]{Marcello Farina}
\author[First]{Heejin Ahn} 

\address[First]{School of Electrical Engineering, Korea Advanced Institute of Science and Technology, Daejeon, Korea (e-mail: \{jehyeophan@kaist.ac.kr, ministop, heejin.ahn\}@kaist.ac.kr).}
\address[Second]{Dipartimento di Elettronica e Informazione, Politecnico di Milano, Milano, Italy, (e-mail: \{alessandro.colombo, marcello.farina\}@polimi.it)}

\begin{abstract}
As connected and autonomous driving technologies advance, vehicles increasingly rely on data from external sensors. Although this information can enhance state estimation, processing all available streams imposes significant communication and computational costs. To address this challenge, we introduce a Sensor Management Center (SMC) that selects a low-cost subset of external sensors in real time while satisfying chance-constrained error bounds derived from an Extended Kalman Filter (EKF) covariance. We formulate the selection problem as a multidimensional minimum knapsack problem and adopt a deficiency-weighted greedy algorithm as an approximate yet efficient solution. The proposed approach is validated through MATLAB simulations and experiments on a 1:15-scale cooperative driving testbed.
\end{abstract}


\end{frontmatter}
%===============================================================================

\section{Introduction}

%background
With advances in autonomous driving technologies, the development of Vehicle-to-Everything (V2X) communication has enabled autonomous vehicles to exchange information with surrounding agents, such as other vehicles and roadside units (RSUs). These advances expose autonomous vehicles to a large volume of external sensory data.

%problems
Although such information can enhance estimation and perception performance, processing all available data streams introduces substantial communication and computational overhead. In addition, participating RSUs and vehicles have little incentive to share sensor data without compensation. These issues highlight the need for a coordinating entity that manages sensor data sharing among networked RSUs and vehicles while guaranteeing estimation performance and fair compensation for contributing agents.

In this paper, we propose a Sensor Management Center (SMC) framework to coordinate cost-aware sensor data selection in V2X environments. When a vehicle issues a service request, the SMC aggregates external sensor data and their associated costs from participating agents that observe the target vehicle. To satisfy the estimation accuracy required by the requester, the SMC selects a minimum cost sensor set.

%related works 
Existing studies on sensor selection have mainly focused on minimizing estimation error under resource constraints. For offline scheduling, \cite{Vitus2012} proposed a pruning-based method that exploits the monotonicity of the Riccati recursion to remove suboptimal sensor combinations, while \cite{maity2022sensor} formulated the scheduling problem as a convex optimization problem and obtained a suboptimal schedule via covariance tracking. For online selection, \cite{wu2020optimal} studied sensor scheduling to reduce the estimation error under bandwidth constraints. Similarly, \cite{yang2023sensor} minimized the estimation error by selecting sensors under Quality-of-Service (QoS) constraints.

This paper takes the complementary perspective of selecting a minimum cost sensor set while satisfying a prescribed estimation accuracy. To make this problem suitable for real-time implementation, we propose a two-stage sensor selection framework. The first stage reformulates the sensor selection problem as a knapsack-type optimization problem by deriving selection criteria from the covariance information produced by the Extended Kalman Filter (EKF)~\citep{anderson2005optimal}. We then solve the resulting problem using a greedy algorithm~\citep{pisinger2004knapsack}. The second stage computes probabilistic error bounds from the updated covariance matrix after selection. We extend prior cost-minimizing sensor selection work~\citep{ahn2019moving} toward a greedy knapsack-type method with post-selection error-bound certification. We validate its effectiveness and real-time feasibility through MATLAB simulations and hardware experiments on a 1:15-scale cooperative driving testbed.

%=============================================================================================================================
\section{Sensor Management System}\label{section2:motivation}

\begin{figure}[t]
    \centering
    \includegraphics[width=\linewidth]{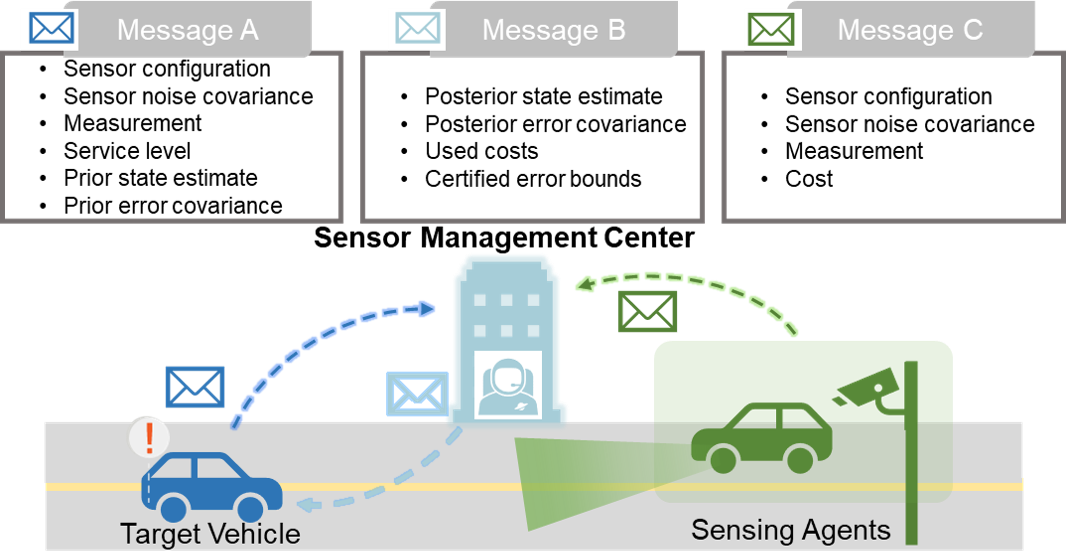}
    \caption{Overview of the Sensor Management Center}
    \label{fig:Framework}
\end{figure}

%===============================================================================
\subsection{Motivation}
Although connected vehicles can handle most driving situations using their onboard sensors, they may require higher estimation accuracy in challenging conditions, such as when sensor noise increases significantly (e.g., due to degraded GPS signals or occluded environments). In such cases, relying solely on onboard sensing may not provide sufficiently reliable state estimates, so vehicles should consider leveraging measurements from external sources, such as roadside sensors or nearby vehicles. 

The Sensor Management Center (SMC) proposed in this paper supports this functionality by collecting external measurements, selecting suitable sensors based on predefined service levels (error bounds) and cost (e.g., sensor usage and compensation costs), and providing the resulting state estimate, as shown in Fig.~\ref{fig:Framework}.
%===============================================================================
\subsection{Sensor Management Center}
The sensor management service provided by the SMC is described in Algorithm~\ref{alg:framework}. When the target vehicle requires improved estimation accuracy, it sends a service level request to the SMC together with its onboard sensor information (i.e., sensor configuration, noise covariance, and measurements) packaged in Message A. Upon receiving the request, the SMC acquires external sensor information and their associated costs, provided in Message C by surrounding agents (e.g., RSUs or nearby vehicles) that are capable of observing the target vehicle.

Using this collected information, the SMC executes the \textsc{Sensor Selection} algorithm (i.e., Algorithm~\ref{alg:selection_procedure} in Section~\ref{section4: solution}) to determine an optimal subset of sensors that satisfies the requested service level while minimizing the total cost. The resulting state estimate and the associated total cost are then transmitted back to the target vehicle via Message B. This estimation service is performed iteratively until the target vehicle terminates the service.

\begin{algorithm}
\caption{Sensor Management Service}
\label{alg:framework}
\begin{algorithmic}[0]
\State \textbf{Start:} A target vehicle \textbf{sends} \emph{Message A} to the SMC.               
\State The SMC \textbf{acknowledges} and \textbf{opens} a session.
\While{Session is open}
  \State Collect \emph{Message C} from sensing agents.
  %\State SMC \textbf{checks} if improved estimation is still needed.
  \State \emph{Message B} = \textsc{Sensor Selection}(\emph{Message A $\&$ C}).
  \State Send \emph{Message B} to the target vehicle.
  \State \textbf{If} vehicle ends service \textbf{then break}.
\EndWhile
\end{algorithmic}
\end{algorithm}

Although a target vehicle may choose to reject the service, for instance, if the total cost exceeds its acceptable threshold, we omit such cases here for simplicity of exposition. Nevertheless, the proposed framework can be readily extended to accommodate such scenarios.

%=============================================================================================================================
\section{Sensor Selection Problem}\label{section3: problem_statement}
In this section, we first describe the target vehicle dynamics, state estimation preliminaries, and service level constraints. We then formulate our sensor selection problem as a multidimensional knapsack optimization.

%===============================================================================
\subsection{Target Vehicle Dynamics} 
We consider a discrete-time nonlinear dynamic model of the target vehicle, expressed as
\begin{equation}
%x_{t+1} = f(x_t, u_t) + B_w w_t, \qquad
x_{t+1} = f(x_t, u_t) + w_t, \qquad
y_t = C_t x_t + v_t,
\label{eq:nonlinear_model}
\end{equation}
where $x_t \in \mathbb{R}^{n_x}, u_t\in \mathbb{R}^{n_u},$ and $y_t\in \mathbb{R}^{n_y}$ denote the system state, control input, and measurement at time step $t$, respectively. The process noise $w_t \sim \mathcal{N}(0, Q_t)$ and measurement noise $v_t \sim \mathcal{N}(0, V_{C_t})$ are assumed to be mutually independent zero-mean Gaussian random variables with uniformly bounded covariance matrices ${Q}_t$ and ${V}_{C_t}$, respectively. %The disturbance matrix $B_w \in \mathbb{R}^{n_x \times n_w}$ 
The measurement matrix $C_t$ represents the concatenation of all sensors selected at time $t$, including the target vehicle’s onboard sensors and, when SMC assistance is active, the external sensors selected by the SMC.

%===============================================================================
\subsection{State Estimation}
Given the measurement vector $y_t$, the SMC estimates the target vehicle state at each time step. Let $\hat{x}_{t|t}$ denote the \emph{posterior} (updated) state estimate at time $t$, and $\hat{x}_{t|t-1}$ denote the \emph{prior} (predicted) state estimate for time $t$. The posterior and prior estimation error covariances are denoted by $P_{t|t}$ and 
$P_{t|t-1}$, respectively.

State estimation is performed using the EKF by linearizing the nonlinear dynamics around the current state estimate. The EKF consists of two main stages: the \textit{prediction} step and the \textit{update} step.
%We use the EKF for state estimation. The EKF handles nonlinear system dynamics by linearizing them around the current state estimate when computing the Kalman gains and covariance updates.
%The EKF consists of two main stages: the \textit{prediction} step and the \textit{update} step.

\textit{Prediction step.}
Based on the system model~\eqref{eq:nonlinear_model}, the prior state estimate and its error covariance are computed as
\begin{align}
    \hat{x}_{t|t-1} &= f(\hat{x}_{t-1|t-1}, u_{t-1}), \label{eq:ekf_pred_state}  \\
    %P_{t|t-1} &= F_{t-1} P_{t-1|t-1} F_{t-1}^\top + B_w Q_{t-1} B_w^\top, \label{eq:ekf_pred_cov}
    P_{t|t-1} &= F_{t-1} P_{t-1|t-1} F_{t-1}^\top + Q_{t-1}, \label{eq:ekf_pred_cov}
\end{align}
where $F_{t-1} = \frac{\partial f}{\partial x} \big|_{\hat{x}_{t-1|t-1}, {u}_{t-1}}$ is the Jacobian of the process model with respect to the state.

\paragraph*{Update step.}
When a new measurement ${y}_{t}$ becomes available, the posterior state estimate is updated as
\begin{align}
%K_{t} &= P_{t|t-1} C_{t}^\top (C_{t} P_{t|t-1} C_{t}^\top + V_{C_{t}})^{-1}, \\
\hat{x}_{t|t} &= \hat{x}_{t|t-1} + K_{t}(y_{t} - C_{t}\hat{x}_{t|t-1}), \label{eq:ekf_x_update}\\
P_{t|t} &= (I - K_{t} C_{t}) P_{t|t-1},
\label{eq:ekf_cov_update}
\end{align}
where $K_{t} = P_{t|t-1} C_{t}^\top (C_{t} P_{t|t-1} C_{t}^\top + V_{C_{t}})^{-1}$ is the Kalman gain.
%where $K_{t}$ is called the Kalman gain.

The covariance update in~\eqref{eq:ekf_cov_update} can also be written in the information form as
\begin{equation}
    P_{t|t}^{-1} = P_{t|t-1}^{-1} + C_t^\top V_{C_t}^{-1} C_t.
    \label{eq:ekf_info_update}
\end{equation}

Through this recursive process, the EKF provides real-time state estimates and covariance updates for the target vehicle based on the available sensor data.

%=============================================================================== 
\begin{assum}[EKF initialization]
\label{asm:ekf_init}
At the initial time $t = 0$, the error covariance 
$P_{0|0}\succ 0$ is known and finite. 
\end{assum}

%\begin{assumption}[Baseline detectability]
\begin{assum}[Onboard sensor detectability]
\label{asm:baseline_detect}
Let $C_t^{\mathrm{on}}$ denote the measurement matrix of the target vehicle's onboard sensors. The linearized pairs $(F_t, C_t^{\mathrm{on}})$ are uniformly detectable for all $t$. See~\cite{anderson1981detectability}.
\end{assum}

%=============================================================================== 
Under Assumptions~\ref{asm:ekf_init} and~\ref{asm:baseline_detect}, standard results for linear time-varying Kalman filters imply that the corresponding error covariances 
$P_{t|t}$ remain uniformly bounded and positive definite (see Corollary~5.4 in~\cite{anderson1981detectability}). 
Because we use the EKF recursion \eqref{eq:ekf_pred_state}--\eqref{eq:ekf_cov_update}, which coincides with the linear time-varying Kalman filter applied to locally linearized models of \eqref{eq:nonlinear_model}, our EKF covariance sequence $P_{t|t}$ remains uniformly bounded and positive definite. The EKF covariances can be interpreted as local approximations of the true error covariances.

%=============================================================================== 
\subsection{Service Level Constraint}
The SMC selects, at each time step, a set of available sensors that yields a state estimate whose error lies within the bound associated with the service level requested by the target vehicle. A finite set of service levels $\mathcal{L}=\{1,\dots,m\}$ is predefined, 
and the requested service level $\ell_t\in\mathcal{L}$ may vary over time.
Each level $\ell_t$ is associated with an error bound $\Omega(\ell_t)\subset \mathbb{R}^{n_x}$. 

Although the estimation error is not truly Gaussian in the nonlinear setting, we rely on the linear approximation previously introduced through EKF and we approximate it as a Gaussian stochastic process with covariance $P_{t|t}$ from EKF.
Since Gaussian estimation errors are inherently unbounded, it is impossible to guarantee deterministic constraints $e_t \in \Omega(\ell_t)$ where $e_t$ is the estimation error. Instead, we enforce the following chance constraint
\begin{equation}
    \Pr(e_t \in \Omega(\ell_t)) \ge p,
    \label{eq:chance_constraint}
\end{equation}
Under the already-discussed Gaussian error approximation,~\eqref{eq:chance_constraint} is fulfilled, according to ~\cite{ahn2019moving}, if
\begin{equation}
    \mathcal{E}(t) \subseteq \Omega(\ell_t),
    \label{eq:set_inclusion}
\end{equation}
where
\begin{equation}
    \mathcal{E}(t) := 
    \{\, e : e^\top P_{t|t}^{-1} e \le \alpha(p)\,\},
    \label{eq:ellipsoid_region}
\end{equation}
and $\alpha(p)$ is obtained from the $\chi^2$ distribution with
$n_x$~degrees of freedom, i.e., 
$\Pr\!\left(\sum_{i=1}^{n_x} Z_i^2 \le \alpha(p)\right)=p$
for $Z_i\sim\mathcal{N}(0,1)$.  %\citep{ahn2019moving}.

Unlike the error bound formulation in \cite{ahn2019moving}, we consider a more specific case in which the error bound $\Omega(\ell_t)$ is represented as an axis-aligned box. In view of this, we derive scalar inequalities implied by the set-inclusion \eqref{eq:set_inclusion} along each axis.  

\setcounter{thm}{0}
\begin{prop}\label{prop:set_inclusion}
    Assume that set $\Omega(\ell_t)$ is defined by
\begin{equation}\label{eq:Omega_rect}
    \Omega(\ell_t) := \{\, e : |h_i^\top e| \le k_i^{(\ell_t)},\; i=1,\dots,n_x \},
\end{equation}
where $\{h_i\}$ denote the canonical basis vectors in $\mathbb{R}^{n_x}$ and $k_i^{(\ell_t)}$ denotes the admissible distance from the origin along the axis. 
If the set-inclusion constraint~\eqref{eq:set_inclusion} holds, then
\begin{equation}
    (P_{t|t}^{-1})_{ii} \;\ge\;
    \frac{\alpha(p)}{\bigl(k_i^{(\ell_t)}\bigr)^2},
    \qquad \forall i=1,\dots,n_x,
    \label{eq:scalar_equiv}
\end{equation}
where $(P_{t|t}^{-1})_{ii}$ is the $i$-th diagonal element of $P_{t|t}^{-1}$.
\end{prop}

\begin{pf}
From \cite{ahn2019moving}, the scalar condition
\begin{equation}
h_i^\top P_{t|t} h_i \le k_i^{(\ell_t)2}/\alpha(p), \qquad i=1,\dots,n_x,
\label{eq:equiv_condition_Ahn}
\end{equation}
is equivalent to the set–inclusion constraint
\eqref{eq:set_inclusion} given~\eqref{eq:ellipsoid_region} and ~\eqref{eq:Omega_rect}. Applying the matrix Cauchy--Schwarz inequality gives $|(P_{t|t}^{-1/2}h_i)^\top (P_{t|t}^{1/2}h_i)|^2  =1\le (h_i^\top P_{t|t}^{-1} h_i) (h_i^\top P_{t|t} h_i)$. If $h_i^\top P_{t|t} h_i\le k_i^{(\ell_t)2}/\alpha(p)$, then we have $h_i^\top P_{t|t}^{-1}h_i \geq \alpha(p)/k_i^{(\ell_t)2}$, which is analogous to~\eqref{eq:scalar_equiv}. ~~~~~~~~~~~~~~\qed
\end{pf}

Although \eqref{eq:scalar_equiv} is only a necessary condition, we use \eqref{eq:scalar_equiv} in our sensor selection problem because the inverse covariance update \eqref{eq:ekf_info_update} explicitly reveals the benefit of adding new sensors. By contrast, the equivalent scalar condition \eqref{eq:equiv_condition_Ahn} is written in terms of the covariance (not the inverse), whose update does not directly expose how individual sensors influence the estimation accuracy. 
If all available sensors still cannot satisfy \eqref{eq:scalar_equiv}, the SMC uses the full sensor set and computes the corresponding rectangular bounds based on the resulting posterior covariance (Section~\ref{section4: solution}).

\setcounter{thm}{0}
\begin{rem}
In addition to the necessary condition provided in Proposition~\ref{prop:set_inclusion}, 
we can also provide a sufficient, although possibly conservative, condition for the 
set-inclusion constraint~\eqref{eq:set_inclusion} by using suitable diagonal consistent 
upper bounds of the covariance matrices. In particular, according to~\cite{farina2016partition}, 
it is possible to define, in a recursive fashion, a diagonal matrix 
$\bar{P}_{t|t}$ such that $\bar{P}_{t|t} \succeq P_{t|t}$. 
Then, as in the previous argument, the condition
\[
(\bar{P}_{t|t}^{-1})_{ii} \;\ge\; \frac{\alpha(p)}{(k_i^{(\ell_t)})^2},
\quad i=1,\dots,n_x,
\]
provides a sufficient condition for the scalar inequality~\eqref{eq:scalar_equiv} 
and hence for the set-inclusion constraint~\eqref{eq:set_inclusion}. 
The conservativeness introduced by this diagonal bounding strategy will be 
further analyzed in future work.
\end{rem}

%=============================================================================== 
\subsection{Knapsack-based Sensor Selection Problem}
Assume that $n_s$ sensors are available at time $t$, each with an associated cost $c_t^{(j)}$. Our objective is to select a subset of these sensors that minimizes total cost while ensuring the estimation error remains within the bound specified by the requested service level $\ell_t$. For tractability, we relax the original constraint \eqref{eq:set_inclusion} by using \eqref{eq:scalar_equiv} as follows.
\begin{subequations}\label{eq:minKP}
\begin{align}
\min_{s_j, \forall j} \quad 
& \sum_{j=1}^{n_s} c_t^{(j)} s_j, \label{eq:minKP_a} \\[4pt]
\text{s.t.} \quad
& \left[ \left[P_{t|t-1}^{-1}\right]_{ii}
  + \sum_{j=1}^{n_s} \left[C_t^{(j)\top} V_{C_{t}^{(j)}}^{-1} C_t^{(j)}\right]_{ii} s_j \right]
  \ge \frac{\alpha(p)}{\bigl(k_i^{(\ell_t)}\bigr)^2}, \label{eq:minKP_b}\\[-0.5ex]
\notag &\text{for } i = 1,2,\dots,n_x, \\[4pt]
& s_j \in \{0,1\}, \quad j = 1,2,\dots,n_s. \label{eq:minKP_c}
\end{align}
\end{subequations}
where $s_j$ is the binary decision variable that takes a unitary value if and only if the $j$th sensor is selected. For the $j$th sensor, $C_t^{(j)}$ denotes its observation matrix and $V_{C_t^{(j)}}$ represents the corresponding measurement noise covariance. The constraint~\eqref{eq:minKP_b} is based on \eqref{eq:scalar_equiv} and the covariance recursion \eqref{eq:ekf_info_update}. The explicit dependence of \eqref{eq:scalar_equiv} on $s_j$ enables its interpretation as a knapsack optimization problem. Each sensor contributes information $[C_t^{(j)\top} V_{C_{t}^{(j)}}^{-1} C_t^{(j)}]_{ii}$ at a cost $c_t^{(j)}$, and the total information must exceed a prescribed threshold. 

To be more specific, our sensor selection problem \eqref{eq:minKP} corresponds to the multidimensional minimum knapsack problem. Indeed, while the classical knapsack problem seeks to maximize the total value of selected items subject to a weight capacity, the minimum knapsack problem is its dual form, aiming to minimize the total weight while achieving a required minimum total value~\citep{csirik1991heuristics}. The multidimensional knapsack problem (d-KP) \citep{pisinger2004knapsack} extends this formulation by incorporating multiple resource constraints. 

%=============================================================================================================================
\section{Solution to the Sensor Selection Problem}\label{section4: solution}
To solve~\eqref{eq:minKP}, we rely on the solution strategies commonly used for the d-KP formulation. The d-KP problem is addressed using exact algorithms, such as Branch-and-Bound and Dynamic Programming, and heuristics, such as greedy strategies and relaxation-based heuristics. Exact methods guarantee optimality but are intractable as the number of sensors increases. Therefore, we adopt a greedy algorithm for d-Min-KP to ensure real-time feasibility.
%=============================================================================== 
\subsection{Greedy Approach for d-Min-KP}
In this work, we adopt a greedy strategy developed for the d-KP problem and adapt it to the d-Min-KP formulation. 
First, note that~\eqref{eq:minKP} can be formulated as
the multidimensional minimum knapsack problem (d-Min-KP):
\begin{equation}
\label{eq:dkp_basic}
\min_{s_j} \sum_{j=1}^{n_s} w_j s_j, 
\quad \text{s.t. } \sum_{j=1}^{n_s} v_{ij} s_j \ge b_i, \; i = 1, \dots, d,
\end{equation}
where $w_j$ is the weight of item $j$, 
$s_j \in \{0,1\}$ indicates whether item $j$ is selected, $v_{ij}$ is the value contribution of item $j$ in the $i$th resource, and $b_i$ is the required value threshold in the $i$th resource. The total number of resource dimensions is $d$. In our setting, the information contribution of each sensor and its associated cost correspond to the item value $v_{ij}$ and item weight $w_j$, respectively, and the required bound maps to $b_i$ in \eqref{eq:dkp_basic}.
Thus, we define the efficiency $e_j$ of each item $j$ and the relative deficiency $r_i$ of each resource $i$ as follows.
\begin{align}
e_j := \frac{\sum_{i=1}^{d} r_i v_{ij}}{w_j}, 
&& 
r_i := \frac{d_i}{b_i^2},
\end{align}
where $d_i:=b_i-\sum_{\{j: s_j=1\}}v_{ij}$ is the remaining information deficiency of resource $i$. Thus, $r_i$ reflects how scarce the remaining information is for resource $i$, and $e_j$ represents how effective item $j$ is at improving the remaining information per unit cost. The greedy algorithm then iteratively selects the item with the highest efficiency, prioritizing those that most effectively reduce the remaining deficiency.

\begin{algorithm}
\caption{\textsc{Sensor Selection}}
\label{alg:selection_procedure}
\begin{algorithmic}[1]
\State \textbf{Input:} 
\State \quad \textit{Message A} $= \big(C_t^{\mathrm{on}},\, V_{C_t^{\mathrm{on}}},\, y_t^{\mathrm{on}},\, \hat{x}_{t|t-1}, \, P_{t|t-1}, \,\ell_t\big)$
\State \quad \textit{Message C} $= \big\{ \big(C_t^{(j)},\, V_{C_t^{(j)}},\, y_t^{(j)},\, c_t^{(j)}\big) \big\}_{j=1}^{n_s}$
\vspace{0.3em}
\State \textbf{Initialization:}
%\State \quad Construct $\mathcal{S}^{\mathrm{all}} = \{ (C_t^{(j)}, V_{C_t^{(j)}}, c_t^{(j)}, y_t^{(j)}) \}_{j=1}^{n_s}$
\State \quad Let $P_{\mathrm{base}} = P_{t|t-1}^{-1} + (C_t^{\mathrm{on}})^\top V_{C_t^{\mathrm{on}}}^{-1} C_t^{\mathrm{on}}$
\State \quad Compute $b_i = \alpha(p)/(k_i^{(\ell_t)})^2- [P_{\mathrm{base}}]_{ii}$ for all $i$
\State \quad Compute $v_{ij} = [(C_t^{(j)})^\top V_{C_t^{(j)}}^{-1} C_t^{(j)}]_{ii}$ for all $i,j$
\State \quad Set $\mathcal{S}_{\textrm{select}} = \emptyset$ and $\mathcal{S}_{\mathrm{temp}} = \{1,\ldots, n_s\}$

\vspace{0.3em}
\State \textbf{Selection: Repeat}
    \State \quad $d_i=b_i-\sum_{j\in \mathcal{S}_{\textrm{select}}} v_{ij}$ and $r_i = d_i / (b_i)^2$ for all $i$
    \State \quad $e_j = \sum_{i=1}^{d} r_i \, v_{ij}/{c_j}$ for all $j \in \mathcal{S}_{\mathrm{temp}}$
    \State \quad$j^* = \arg\max_{j\in\mathcal{S}_{\mathrm{temp}}}  e_j $\label{alg:argmax}
    \State \quad$\mathcal{S}_{\textrm{select}} = \mathcal{S}_{\textrm{select}} \cup \{ j^* \}$, \quad $\mathcal{S}_{\mathrm{temp}} = \mathcal{S}_{\mathrm{temp}} \setminus \{ j^* \}$ \label{alg:S_select}
    \State \quad \textbf{If} all $d_i -v_{ij^*}\le 0$ or $\mathcal{S}_{\mathrm{temp}}=\emptyset$ \textbf{then}
    \State \qquad $\mathcal{S}^* = \mathcal{S}_{\textrm{select}}$; \textbf{break}

\vspace{0.3em}
\State \textbf{Post-processing:}
\State \quad Update $P_{t|t}$ and $\hat{x}_{t|t}$ using \eqref{eq:ekf_cov_update}, \eqref{eq:ekf_x_update} with $\mathcal{S}^*$
\State \quad Compute $c_t^*=\sum_{j\in \mathcal{S}^*} c_t^{(j)}$ and $k_i^* = \sqrt{\alpha(p)\, h_i^\top P_{t|t} h_i}$
\State \textbf{Return:} \textit{Message B}$=(\hat{x}_{t|t}, P_{t|t}, c_t^*, k_i^*~ \forall i)$
\end{algorithmic}
\end{algorithm}

Algorithm~\ref{alg:selection_procedure} summarizes our greedy sensor selection procedure. The algorithm begins by receiving the set of available sensor information from Messages~A and~C, which includes the observation models, covariances, measurements, a prior state estimate and error covariance, and costs of all candidate sensors along with a service level. Using the prior $(P_{t|t-1})^{-1}$ together with the onboard sensors' contribution $(C_t^{\mathrm{on}})^\top V_{C_t^{\mathrm{on}}}^{-1} C_t^{\mathrm{on}}$, we construct the base covariance matrix $P_{\mathrm{base}}$. We then compute the thresholds $b_i$ and the item values $v_{ij}$.

The algorithm then iteratively selects the sensor with the highest efficiency 
$e_j$ and inserts it into $\mathcal{S}_{\mathrm{select}}$. The iteration 
terminates when either all deficiency components satisfy 
$d_i - v_{ij^*} \le 0$ (after adding sensor $j^*$), or no additional sensors 
remain to be added (i.e., $\mathcal{S}_{\mathrm{temp}} = \emptyset$).

Finally, the selected sensor set $\mathcal{S}^*=\mathcal{S}_{\textrm{select}}$ is used to perform the EKF update, resulting in a posteriori covariance $P_{t|t}$ and the state estimate $\hat{x}_{t|t}$. We compute the total consumed cost $c_t^*$ and the rectangular bounds $k_i^*$ using $P_{t|t}$, as defined in Theorem~\ref{thm:posterior_rect}. This post-processing step is necessary because Algorithm~\ref{alg:selection_procedure} does not guarantee that the service level constraint \eqref{eq:set_inclusion} is satisfied, as noted in Proposition~\ref{prop:set_inclusion}. All estimation results are then packaged into Message~B and sent back to the target vehicle.

%=============================================================================== 
\subsection{Main result}
The following theorem provides the main properties of the selected sensor set.

\setcounter{thm}{0}
\begin{thm}[A posteriori rectangular bound]
\label{thm:posterior_rect}
Let $\mathcal{S}^*, P_{t|t}^*$, and $k_i^*, i=1,\ldots,n_x$ denote the selected sensor set, the posterior covariance, and the rectangular bounds obtained in the post-processing step of Algorithm~\ref{alg:selection_procedure} for a requested service level $\ell_t$. 
Define
$d_i^* := b_i - \sum_{j \in \mathcal{S}^*} v_{ij}, i = 1,\dots,n_x$ and $\mathcal{E}^*(t) := \left\{\, e : e^\top (P_{t|t}^*)^{-1} e \le \alpha(p) \,\right\}.$
Then the following statements hold:
\begin{enumerate}
    \item If $d_i^* > 0$ for some $i$, then no sensor subset can satisfy the set-inclusion constraint~\eqref{eq:set_inclusion} for the requested service level $\ell_t$.
    
    \item If $d_i^* \le 0$ for all $i$, then the selected sensor set $\mathcal{S}^*$ satisfies the scalar inequalities~\eqref{eq:minKP_b} but not necessarily the set-inclusion constraint~\eqref{eq:set_inclusion}.
    
    \item For $\Omega^*(t) := \left\{\, e : |h_i^\top e| \le k_i^*,\; i=1,\dots,n_x \,\right\},$
    we have $\mathcal{E}^*(t) \subseteq \Omega^*(t)$, and $\Omega^*(t)$ is the smallest axis-aligned rectangular error bound that contains $\mathcal{E}^*(t)$.
\end{enumerate}
\end{thm}

\begin{pf}
By definition, $d_i^* = \alpha(p)/(k_i^{(\ell_t)})^2 - \bigl([P_{\mathrm{base}}]_{ii} + \sum_{j\in\mathcal{S}^*} v_{ij}\bigr)$.
\begin{enumerate}
    \item If $d_i^* > 0$ for some $i$, then
$[P_{\mathrm{base}}]_{ii} + \sum_{j\in\mathcal{S}^*} v_{ij}
  < \alpha(p)/(k_i^{(\ell_t)})^2$, so the scalar inequality
\eqref{eq:minKP_b} fails for that $i$. By the contrapositive of
Proposition~\ref{prop:set_inclusion}, the set-inclusion constraint
\eqref{eq:set_inclusion} for the requested level $\ell_t$ cannot hold.
\item If $d_i^* \le 0$ for all $i$, then
$[P_{\mathrm{base}}]_{ii} + \sum_{j\in\mathcal{S}^*} v_{ij}
  \ge \alpha(p)/(k_i^{(\ell_t)})^2$ for all $i$, so
\eqref{eq:minKP_b} holds. However, by Proposition~\ref{prop:set_inclusion}, this is only
a necessary condition for \eqref{eq:set_inclusion}.
\item For any error bound $k_i$, the inequality $h_i^\top P_{t|t}^* h_i \leq (k_i)^2 / \alpha(p)~\forall i$ is equivalent to the set inclusion $\mathcal{E}^*(t) \subseteq \{e:|h_i^\top e| \le k_i, \forall i \}$. The equality holds at $k_i=k_i^*$ because in the post-processing step, we set $k_i^* := \sqrt{\alpha(p)\, h_i^\top P_{t|t}^* h_i}$ and thus $h_i^\top P_{t|t}^* h_i = (k_i^*)^2 / \alpha(p), \forall i$. Therefore, $\mathcal{E}^*(t) \subseteq \Omega^*(t)$ and $\Omega^*(t)$ is the smallest rectangular bound.~~~~~~~~~~~~~~~~~~~~~~~~~~~~~~~~~~~~~~~~~~\qed
\end{enumerate}
\end{pf}

%=============================================================================================================================
\section{Experiments}\label{section5: experiments}
We implement the framework in MATLAB R2023a using YALMIP and MOSEK to solve~\eqref{eq:minKP}. Hardware experiments use the 1:15-scale testbed from~\cite{bae2025miniaturetestbedvalidatingmultiagent}, featuring three vehicles and one RSU with a depth camera.

%=============================================================================== 
\subsection{MATLAB Simulation}
A vehicle travels on a road where ten RSUs are measuring its state.  
The SMC uses the nonlinear dynamic model in~\eqref{eq:nonlinear_model} with the vehicle state update function $f(x_t,u_t)$ as:
\begin{align}
    p^x_{t+1} &= p^x_{t} + \nu_t \cos(\psi_t)\Delta t,   &&  p^y_{t+1} = p^y_{t} + \nu_t \sin(\psi_t)\Delta t, \notag\\
    \psi_{t+1} &= \psi_t + \frac{\nu_t}{L}\tan(\delta_t)\Delta t, &&    \nu_{t+1} = \nu_t + a_t \Delta t,\label{eq:CAV_dynamics}
\end{align}

where $p^x_t$ and $p^y_t$ denote the planar position, $\psi_t$ is the yaw angle,
$\nu_t$ is the longitudinal velocity, $a_t$ is the longitudinal acceleration, $\delta_t$ is the steering angle, the wheelbase length $L$ is $3~\mathrm{m}$, and $\Delta t = 0.05~\text{s}$. The state and input of the target vehicle are $x_t = (p^x_t,\, p^y_t,\, \psi_t,\, \nu_t)^\top$ and $u_t = (\delta_t,\, a_t)^\top$, respectively. The process noise $w_t$ is modeled as zero-mean Gaussian with covariance $Q=\mathrm{diag}(0.0031\,\mathrm{m}^2, 0.0031\,\mathrm{m}^2, 0.0001\,\mathrm{rad}^2, 0.0125\,(\mathrm{m/s})^2)$.

For control, the steering angle $\delta_t$ is provided by a pure-pursuit path-following controller~\citep{snider2009automatic}, and the longitudinal command $a_t$ is generated by a proportional integral (PI) controller,
$a_t = K_p (\nu_t^{\mathrm{ref}} - \nu_t) + K_i I_t$, where $\nu_t^{\mathrm{ref}}$ is the reference velocity and $I_t$ is the integrated velocity tracking error with controller gains $K_p=3$ and $K_i=0.1$.

%For estimation, the vehicle is equipped with three onboard sensors:
The vehicle has three zero-cost onboard sensors: a position sensor ($C^1=(1,0,0,0; 0,1,0,0)$), an IMU ($C^2= (0,0,1,0)$), and a velocity encoder ($C^3=(0,0,0,1)$), and their measurement noise covariances are given by $V_{\mathrm{position}}=\mathrm{diag}(3, 3)\,\mathrm{m}^2, V_{\mathrm{IMU}}=0.5\,\mathrm{rad}^2$, and $V_{\mathrm{velocity}}=2\,(\mathrm{m/s})^2$. The measurement matrix, costs, and measurement noise covariances of each roadside sensor are summarized in Table~\ref{tab:rsu_setup}. 
We run simulations for $20$\,s, setting the initial covariance to $P_0 = 0.05I_4$ and $p=0.95$. The plots exclude the initial $0.5$\,s transient period.

\begin{table}[t]
\centering
\caption{RSU configurations in the simulation}
\label{tab:rsu_setup}
\renewcommand{\arraystretch}{1.1}
\setlength{\tabcolsep}{3.5pt} 
\resizebox{\columnwidth}{!}{
\begin{tabular}{c c c c} 
\toprule
\textbf{RSU ($j$)} & \textbf{Meas. Mat. ($C^i$)} & \textbf{Cost ($c_t^j$)} & \textbf{Meas. Cov. ($V_{\mathrm{RSU},j}$)} \\
\midrule
RSU 1 & $C^3$ & 1.01 & $0.005$ \\
RSU 2 & $C^3$ & 1.10 & $0.0051$ \\
RSU 3 & $C^1; C^2$ & 5.01 & $\mathrm{diag}(0.01,0.01,0.01)$ \\
RSU 4 & $C^1; C^2$ & 5.10 & $\mathrm{diag}(0.01,0.01,0.01)$ \\
RSU 5 & $C^1$ & 2.01 & $\mathrm{diag}(0.01,0.01)$ \\
RSU 6 & $C^1$ & 2.10 & $\mathrm{diag}(0.01,0.01)$ \\
RSU 7 & $C^3$ & 1.20 & $0.005$ \\
RSU 8 & $C^1$ & 3.10 & $\mathrm{diag}(0.01,0.011)$ \\
RSU 9 & $C^1$ & 4.01 & $\mathrm{diag}(0.01,0.01)$ \\
RSU 10 & $C^1$ & 4.10 & $\mathrm{diag}(0.011,0.01)$ \\
\bottomrule
\end{tabular}
} 
\end{table}

\begin{figure}[t]
    \centering
    \includegraphics[width=\linewidth]{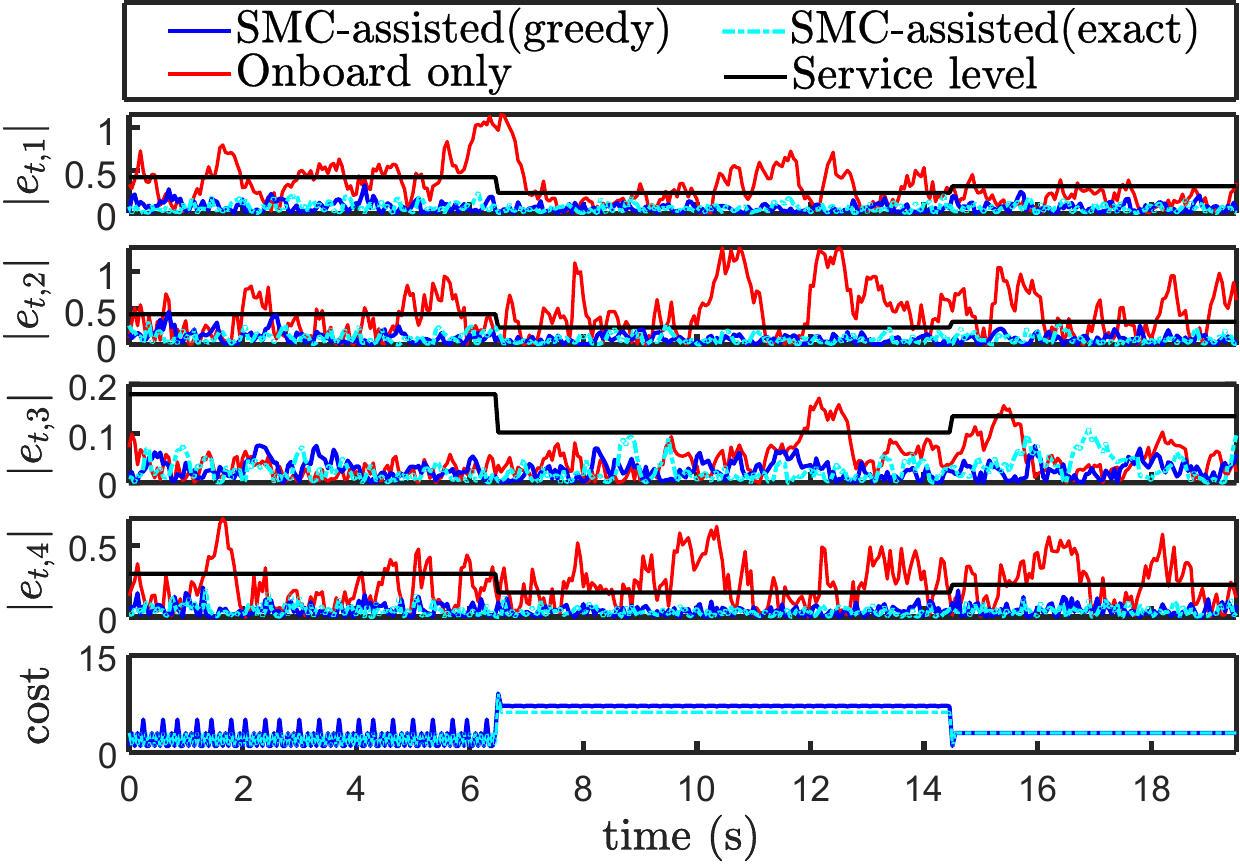}
    \caption{The top four plots show the estimation errors $|e_{t,i}|$ using only onboard sensors (red) and using external sensors with the SMC-assisted (blue and cyan) under three different service levels (black). The bottom plot shows the cost of selected external sensors.}
    \label{fig:Sim_rmse}
\end{figure}

The top four plots in Fig.~\ref{fig:Sim_rmse} show the state estimation errors. 
The red solid, blue solid, and cyan dashed lines represent the results obtained using only onboard sensors, using the greedy selection (Algorithm~\ref{alg:selection_procedure}), and the exact solution of \eqref{eq:minKP}, respectively.

The black solid lines indicate the error bound $k_i^{(\ell_t)}$ corresponding to a requested service level $\ell_t$, where $\ell_t=1$ for $0\leq t< 6.5$, $\ell_t=3$ for $6.5\leq t< 14.5$, and $\ell_t=2$ for $14.5\leq t \leq 19.5$.
Over the entire time horizon, the SMC-assisted cases exhibit significantly improved estimation performance compared to the onboard-only case. In particular, the exact and greedy SMC-assisted solutions select almost the same sensor sets in terms of cost, and their estimation errors are almost indistinguishable. The mean and maximum per-step computation times are $96.2$\,ms and $1675.6$\,ms for the exact solution, and $0.29$\,ms and $4.52$\,ms for the greedy.

\begin{figure}[t]
    \centering
    \includegraphics[width=\linewidth]{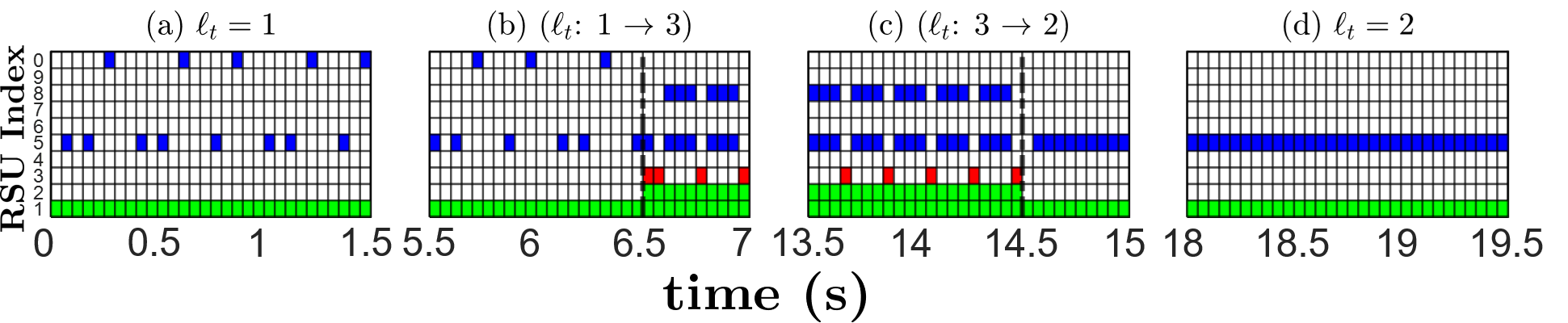}
    \caption{Simulation results. Each row corresponds to one RSU (from RSU~1 at the bottom to RSU~10 at the top). The colors indicate the measurement structure of each RSU: green for $C^3$, blue for $C^{1}$, and red for $(C^{1};C^{2})$.} 
    \label{fig:Selection_result}
\end{figure}

Fig.~\ref{fig:Selection_result} illustrates the RSU sensors selected by the SMC during the simulation. Plots (a) through (d) correspond to portions of the time horizon illustrating a constant service level of $\ell_t=1$, a transition from $\ell_t=1$ to $3$, a transition from $\ell_t=3$ to $2$, and a constant level of $\ell_t=2$, respectively. The black dashed lines in plots (b) and (c) mark the service level transitions.
The SMC adaptively adjusts the sensor selection to satisfy the requested service level while minimizing the cost. Moreover, the selection results align with the cost profile shown in Fig.~\ref{fig:Sim_rmse}: 
achieving higher service levels, as observed after the transition in Fig.~\ref{fig:Selection_result}(b), requires utilizing more sensors and therefore incurs higher cost.

%=============================================================================== 
\subsection{Scaled Testbed Experiment}

\begin{figure}[h]
    \centering
    \includegraphics[width=0.9\linewidth]{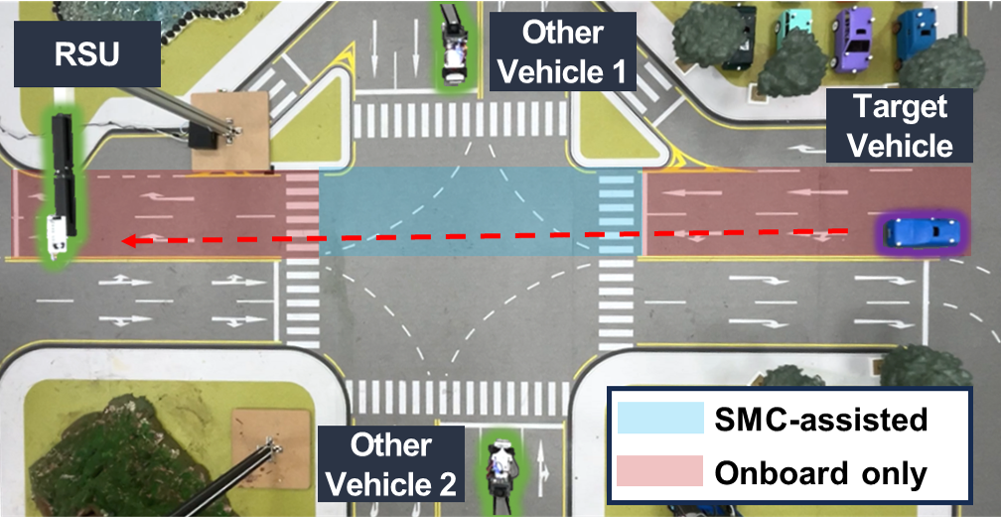}
    \caption{Experiment Scenario. The target vehicle moves along the red-dotted path, and the SMC actively employs the RSU, Other Vehicle 1, and Other Vehicle 2 to satisfy the requested service level.}
    \label{fig:exp_scenario}
\end{figure}

Consider a scenario depicted in Fig.~\ref{fig:exp_scenario}, where a target vehicle experiences a sudden degradation of its position sensor, requiring estimation support from the SMC. In this setup, the external sensing units utilize NVIDIA Jetson Orin NX platforms with RealSense D435 cameras, while the target vehicle operates on a Raspberry Pi4. The SMC runs on a laptop equipped with an Intel Core i7-1165G7 CPU at 2.80 GHz and 16 GB RAM.

The vehicle dynamics follow~\eqref{eq:CAV_dynamics} with $\Delta t = 0.05\,\mathrm{s}$, $L=17\,\mathrm{cm}$, $Q = \mathrm{diag}(0.01\,\mathrm{m}^2,0.01\,\mathrm{m}^2,0.002\,\mathrm{rad}^2 ,0.03\,(\mathrm{m/s})^2)$.

The target vehicle is equipped with an IMU, a wheel encoder, and an OptiTrack motion-capture system. To emulate a noisy onboard GPS, OptiTrack measurements are corrupted with zero-mean Gaussian noise. The measurement covariances are set to $V_{\text{position}} = \text{diag}(0.11^2, 0.11^2)$ m$^2$, $V_{\text{IMU}} = 0.02$ rad$^2$, and $V_{\text{velocity}} = 0.01$ (m/s)$^2$. External sensors estimate position via vision-based depth sensing. The measurement covariance is scaled by $(1+3(1 - \mathrm{conf}))$, where $\mathrm{conf} \in [0,1]$ is the YOLOv8 detection confidence. The base covariances are $0.0012$ for RSU1 and $0.0015$ for Vehicles 1 and 2, with sensing costs of $2$, $1.5$, and $1.501$, respectively.

\begin{figure}[h]
    \centering
    \includegraphics[width=\linewidth]{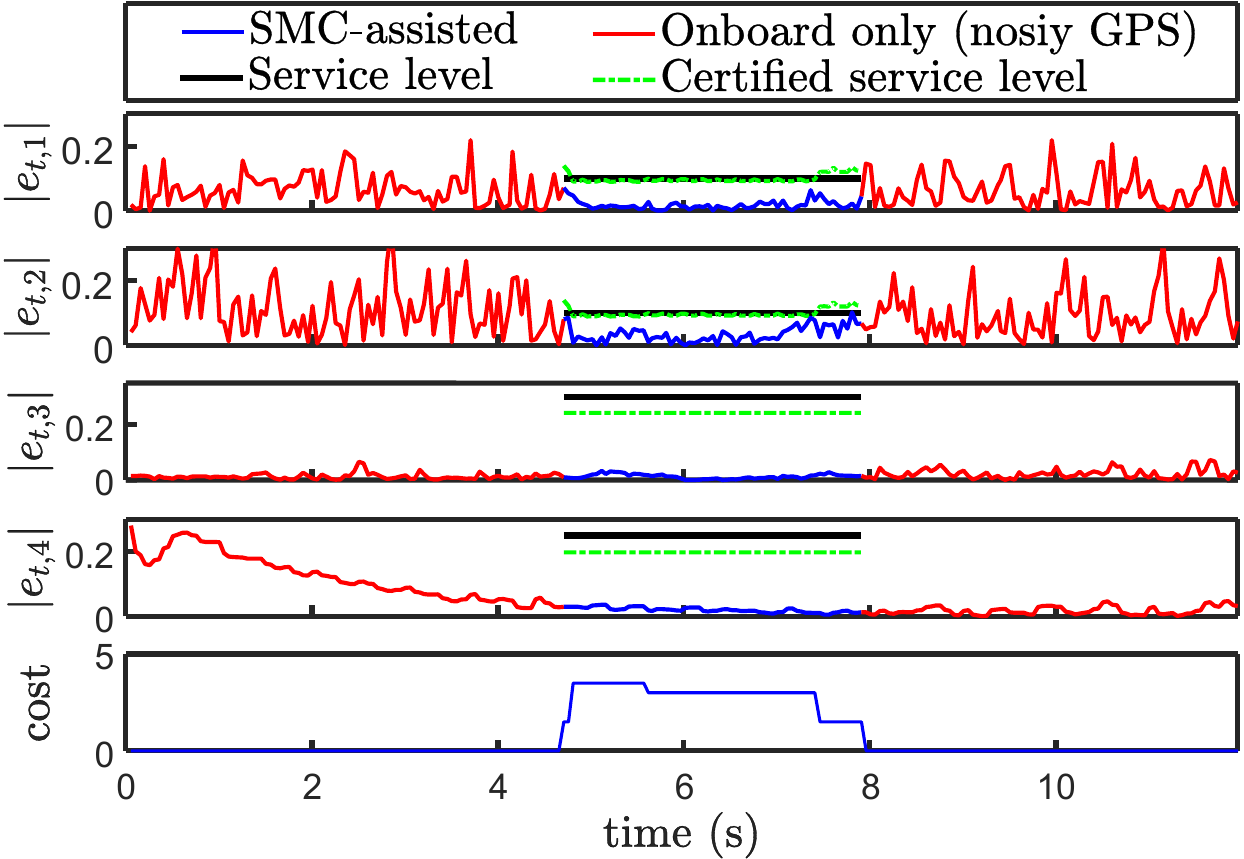}
    \caption{The top four plots show the estimation errors $|e_{t,i}|$. The same color scheme as in Fig.~\ref{fig:Sim_rmse} is used, except the green dashed lines denote the certified service levels $k_i^*$.  The bottom plot shows the cost of selected external sensors.}
    \label{fig:exp_results}
\end{figure}

Fig.~\ref{fig:exp_results} presents the testbed experimental results.
As shown in the top two plots, the estimation errors of $p_t^x$ and $p_t^y$ are large when only onboard sensors are used (red lines). 
Once the SMC-assisted estimation service at service level~$\ell=2$ is activated (blue lines), the estimation errors rapidly decrease and remain within the admissible error bound $\Omega(2)$, computed with $p=0.95$, while the service remains active. The SMC completed both Algorithm~\ref{alg:selection_procedure} and communication of its results within each $50$\,ms control interval. Although our greedy approach does not guarantee the satisfaction of the requested service-level bounds (black lines), the actual estimation errors are often within these bounds. 
In fact, in some intervals, the certified bounds $k_i^*$ (green dashed lines) appear lower than the requested service bound, demonstrating that our greedy algorithm often selects sensors that satisfy the requested bounds.

%=============================================================================================================================
\section{Conclusions}\label{section6: conclusion}
We have proposed a Sensor Management Center framework that performs sensor selection under user-specified service levels. 
We have modeled the estimation-error requirement as a chance constraint with a rectangular error bound and reformulated it, using EKF covariance information, into a multidimensional minimum knapsack problem. Based on this formulation, we have developed a greedy sensor selection algorithm and introduced an a posteriori service-level certification using the updated covariance. MATLAB simulations and 1:15-scale testbed experiments demonstrate that the proposed SMC framework significantly improves estimation accuracy over onboard-only sensing while enabling explicit performance–cost trade-offs. Future work will extend the proposed framework to dynamic, large-scale sensor networks accounting for communication effects, such as latency and packet loss.

\bibliography{ifacconf}             % bib file to produce the bibliography

@article{Vitus2012,
  author    = {Michael P. Vitus and Wei Zhang and Alessandro Abate and Jianghai Hu and Claire J. Tomlin},
  title     = {On Efficient Sensor Scheduling for Linear Dynamical Systems},
  journal   = {Automatica},
  year      = {2012},
  volume    = {48},
  number    = {10},
  pages     = {2482--2493},

}

@article{maity2022sensor,
  title={Sensor scheduling for linear systems: A covariance tracking approach},
  author={Maity, Dipankar and Hartman, David and Baras, John S},
  journal={Automatica},
  year={2022},
  volume={136},
  pages={110078}, 
}

@article{ahn2019moving,
  author={Ahn, Heejin and Danielson, Claus},
  title={Moving Horizon Sensor Selection for Reducing Communication Costs with Applications to Internet of Vehicles},
  journal={2019 American Control Conference (ACC)},
  year={2019},
  pages={1464--1469},
}

@book{anderson2005optimal,
  author={Anderson, Brian DO and Moore, John B},
  title={Optimal filtering},
  publisher={Courier Corporation},
  year={2005},
}

@article{csirik1991heuristics,
  author={Csirik, J{\'a}nos and Frenk, Johannes Bartholomeus Gerardus and Labb{\'e}, Martine and Zhang, Shuzhong},
  title={Heuristics for the 0-1 min-knapsack problem},
  journal={Acta Cybernetica},
  year={1991},
  volume={10},
  number={1-2},
  pages={15--20},
}

@book{pisinger2004knapsack,
  title        = {Knapsack Problems},
  author       = {Hans Kellerer and Ulrich Pferschy and David Pisinger},
  publisher    = {Springer-Verlag Berlin Heidelberg},
  year         = {2004},
}

@article{anderson1981detectability,
  author={Anderson, Brian DO and Moore, John B},
  title={Detectability and stabilizability of time-varying discrete-time linear systems},
  journal={SIAM Journal on Control and Optimization},
  year={1981},
  volume={19},
  number={1},
  pages={20--32},
}

@misc{bae2025miniaturetestbedvalidatingmultiagent,
  author       = {Bae, Hyunchul and Lee, Eunjae and Han, Jehyeop and Kang, Minhee 
                  and Kim, Jaehyeon and Seo, Junggeun and Noh, Minkyun and Ahn, Heejin},
  title        = {Miniature Testbed for Validating Multi-Agent Cooperative Autonomous Driving},
  year         = {2025},
  archivePrefix= {arXiv},
  eprint       = {2511.11022},
  primaryClass = {cs.RO},
  url          = {https://arxiv.org/abs/2511.11022},
  note         = {arXiv preprint}
}

@techreport{snider2009automatic,
  title        = {Automatic Steering Methods for Autonomous Automobile Path Tracking},
  author       = {Snider, Jarrod M.},
  institution  = {Robotics Institute, Carnegie Mellon University},
  year         = {2009}
}

@article{farina2016partition,
  title={Partition-based distributed Kalman filter with plug and play features},
  author={Farina, Marcello and Carli, Ruggero},
  journal={IEEE Transactions on Control of Network Systems},
  volume={5},
  number={1},
  pages={560--570},
  year={2016},
  publisher={IEEE}
}

@article{yang2023sensor,
  title={Sensor selection for remote state estimation with {QoS} requirement constraints},
  author={Yang, Huiwen and Huang, Lingying and Yang, Chao and Mo, Yilin and Shi, Ling},
  journal={Automatica},
  volume={157},
  pages={111241},
  year={2023},
  publisher={Elsevier}
}

@article{wu2020optimal,
  title={Optimal scheduling of multiple sensors over lossy and bandwidth limited channels},
  author={Wu, Shuang and Ding, Kemi and Cheng, Peng and Shi, Ling},
  journal={IEEE Transactions on Control of Network Systems},
  volume={7},
  number={3},
  pages={1188--1200},
  year={2020},
  publisher={IEEE}
}

\end{document}